\documentclass[11pt,twoside]{article} 
\usepackage{ifthen}
\usepackage{soul}
\usepackage{fullpage, prettyref}

\usepackage{fancyhdr}
\usepackage{times}
\usepackage{graphicx}
\usepackage{marvosym}
\usepackage{ifthen}
\usepackage{algorithm}
\usepackage{amsmath, amssymb, amsthm}
\usepackage{thmtools}

\usepackage[noend]{algpseudocode}
\algrenewcomment[1]{\(\triangleright\) {\small{\emph{\color{blue} #1}}}}
\algnewcommand{\LineComment}[1]{\State \(\triangleright\) \emph{\color{blue} #1}}
\algnewcommand{\Invariant}[1]{\State \(\triangleright\) \emph{\color{red} #1}}

\usepackage{url}
\usepackage{boxedminipage}
\usepackage{wrapfig}
\usepackage{ifthen}
\usepackage{color,xcolor}
\usepackage{transparent}

\usepackage{varwidth}

\usepackage{a4,geometry}
\usepackage{graphicx}
\usepackage{amsmath,amssymb,amsthm,mathtools}
\usepackage{bm}
\usepackage{xspace}
\usepackage{url}
\usepackage{boxedminipage}
\usepackage{wrapfig}
\usepackage{ifthen}
\usepackage{color}
\usepackage{xcolor}
\usepackage{framed}
\usepackage{mdframed}
\usepackage[colorlinks=true,pdfpagemode=UseNone,urlcolor=blue,linkcolor=blue,citecolor=violet,pdfstartview=FitH]{hyperref}
\usepackage[noabbrev,nameinlink]{cleveref}
\usepackage{fullpage}
\usepackage{enumitem}
\usepackage{tikz}

\newtheorem{theorem}{Theorem}
\newtheorem{result}{Result}
\newtheorem{lemma}{Lemma}
\newtheorem{claim}{Claim}

\newtheorem{definition}{Definition}

\newtheorem{remark}{Remark}

\newtheorem{observation}{Observation}

\newcommand{\ignore}[1]{}

\newcommand{\cE}{\mathcal{E}}

\newcommand{\cI}{\mathcal{I}}
\newcommand{\cJ}{\mathcal{J}}
\newcommand{\cK}{\mathcal{K}}

\newcommand{\cS}{\mathcal{S}}

\newcommand{\cX}{\mathcal{X}}

\newcommand{\calE}{\mathcal{E}}

\newcommand{\calI}{\mathcal{I}}
\newcommand{\calJ}{\mathcal{J}}
\newcommand{\calK}{\mathcal{K}}

\newcommand{\calP}{\mathcal{P}}

\newcommand{\calS}{\mathcal{S}}

\newcommand{\calX}{\mathcal{X}}

\newcommand{\RR}{\mathbb{R}}
\newcommand{\NN}{\mathbb{N}}

\newcommand{\fI}{\mathfrak{I}}

\newcommand{\frakI}{\mathfrak{I}}

\newcommand{\bx}{\mathbf{x}}
\newcommand{\by}{\mathbf{y}}

\newcommand{\ceil}[1]{{\left\lceil{#1}\right\rceil}}
\newcommand{\floor}[1]{{\left\lfloor{#1}\right\rfloor}}

\newcommand{\Ot}{\ensuremath{\widetilde{O}}}

\newcommand{\abs}[1]{\left\lvert #1 \right\rvert}

\newcommand{\etal}{{\it et al.\,}}

\def\bar{\overline}

\def\2plus{{\tt (++)}}
\def\3plus{{\tt (+++)}}
\def\4plus{{\tt (++++)}}
\def\5plus{{\tt (+++++)}}

{\hspace*{\fill}$\Box$\par}
\newlength{\algobox}
\colorlet{shadecolor}{blue!10}

\newenvironment{rslt}{
	\begin{mdframed}[backgroundcolor=gray!20,topline=false,bottomline=false,leftline=false,rightline=false]\begin{result}
		}{
		\end{result}
	\end{mdframed}
}

\def\eq{\leftarrow}

\def\SUM{\mathsf{SUM}}

\def\CUT{\mathsf{CUT}}
\def\CROSS{\mathsf{CROSS}}

\def\ADDITIVE{\mathsf{ADD}}

\usepackage{times}
\usepackage{paralist}

\def\Choi{{\sc VectorRecovery}}
\def\LiaoC{{\sc WtdSpForestGraph}}
\def\EDGE{\mathsf{EDGE^?}}

\makeatletter
\crefalias{ALC@line}{line} 
\makeatother

\crefname{line}{line}{lines}
\Crefname{line}{Line}{Lines}

\title{Query Complexity of Hypergraph Connectivity and Learnability using CUT Oracles}
\author{Deeparnab Chakrabarty\footnote{Dartmouth College, Email: {\tt deeparnab@dartmouth.edu}} \and Hang Liao\footnote{Palo Alto Networks, Email: {\tt hangliao98@gmail.com}. Work done as a graduate student at Dartmouth College.}}
\date{}

\begin{document}
	
	\maketitle
	
	\begin{abstract}%
		 We investigate the power of CUT queries to reveal the structure of unknown hypergraphs. While simple graphs allow for optimal $O(n)$-query connectivity algorithms, hypergraphs face a fundamental identifiability barrier in that distinct hypergraphs can share identical cut-profiles, making exact edge learning impossible in general, a primitive crucial in the graph connectivity algorithms. 
		 
		 We first present a zero-error randomized algorithm that identifies the connected components of any weighted hypergraph using $O(n)$ expected queries, matching the $\Omega(n)$ lower bound. This approach bypasses the reconstruction barrier by introducing the notion of  ``independent families''---vertex subpartitions that do not share hyperedges---and iteratively coarsening them using auxiliary weighted graph connectivity techniques [Liao-Chakrabarty, 2024]. 
		
		Second, we demonstrate that the impossibility of exact learning depends on hyperedge parity. For even-parity hypergraphs, we show that the structure is reconstructible using a M\"obius transform on the CUT function to implement binary-search-style vertex identification. This yields deterministic algorithms for obtaining $k$-connectivity certificates for $r$-bounded even hypergraphs in $\tilde{O}_r(kn)$ queries. Finally, we bypass parity and rank constraints for linear hypergraphs, achieving a subquadratic $\tilde{O}(kn^{1.5})$ query complexity for $k$-connectivity. This significantly improves upon the general $\tilde{O}(n^2)$ bound derived via symmetric submodular function minimization.

\end{abstract}
\thispagestyle{empty}
\setcounter{page}{0}
\newpage
\section{Introduction}\label{sec:intro}

Motivated by the complexity of symmetric submodular function minimization (SFM),  Rubinstein, Schramm, and Weinberg~\cite{RubinsteinSW18} introduced the $\CUT$-query model to study the connectivity of an undirected graph $G=(V,E)$ whose vertex set on $n$ vertices is known but the edge set is unknown. For any subset $S\subseteq V$, $\CUT(S)$ returns the number/weight of edges crossing from $S$ to $V\setminus S$. A simple binary-search style idea allows one to 
sample a random edge incident on a vertex in $O(\log n)$-many queries. This simple but crucial primitive, along with many other ideas, has been key to many recent results~\cite{MukhopadhyayN20,AssadiCK21,LeeMS21,AuzaL21,ApersEGLMN22,ChakrabartyL23,LiaoC24,AnandSW25,KennethK25,KennethK26a,KennethK26b}; for instance, this has culminated in  
the recent zero-error randomized algorithms~\cite{ApersEGLMN22,LiaoC24}
that determine the connected components of a (possibly weighted) graph making $O(n)$ queries in expectation, and this is tight~\cite{AuzaL21}.

In this paper, we initiate the systematic study of {\em hypergraphs} in the $\CUT$-query model: $\CUT(S)$ now returns the number/weight of hyperedges $e\in \cE$ which intersect both $S$ and $V\setminus S$. While this function still remains a symmetric submodular function -- implying an $\widetilde{O}(n^2)$ query algorithm for finding the global minimum cut -- the transition from graphs to hypergraphs creates a fundamental {\em identifiability barrier}.
Unlike the case of graphs, distinct hypergraphs can have identical $\CUT$ profiles. For example, the $K_4$ graph and the hypergraph with all four 3-element subsets of a 4-element set yield the same cuts; this was explicitly noted in~\cite{ChenKN21} as a bottleneck to obtain cut-sparsifiers in hypergraphs using $\CUT$ queries alone. This inability to learn edges is a hindrance in obtaining low query algorithms in hypergraphs.

\paragraph{\bf Connectivity Structure.} While the naive binary-search style $O(n\log n)$-query algorithm (see, e.g, Theorem 5.1 in~\cite{Harvey08}) to determine connected components works even for hypergraphs,
removing the ``$\log n$'' term is a challenging problem due to the inability of recovering edges due to aforementioned identifiability barrier. Indeed, for graphs random sampling of edges is a key step in the recent $O(n)$-query randomized algorithms of~\cite{ApersEGLMN22,LiaoC24}. Whether one can get an $O(n)$-query algorithm was
was explicitly posed by Chakrabarty and Liao~\cite{ChakrabartyL24} who gave an $O(n)$-query algorithm for the special case of learning partitions which are a very special class of hypergraphs (the hypergraph is a ``hypermatching''). Our first result is an affirmative answer to this question.

\begin{rslt}\label{rslt:O(n)}
	There is a polynomial-time zero-error randomized algorithm that finds the connected components of a weighted hypergraph using $O(n)$ $\CUT$-queries in expectation.
\end{rslt}

\noindent
To overcome the identifiability barrier, we introduce the key notion of {\em independent families} -- a sub-partition of the vertex set which do not share hyperedges across.
In particular, the collection of connected components is an independent family and is the object we desire. Our algorithm maintains a collection of these families, and iteratively 
{\em merges} them so that in $O(\log n)$ rounds, one ends up with the single independent family we desire.  This process is done by reducing this process to a {\em spanning forest learning} process in an auxiliary weighted bipartite (simple) graph, for which we can use the recently 
developed $O(n)$-query graph algorithm of~\cite{LiaoC24}. 
This reduction allows us to bypass the direct learning of hyperedges while correctly maintaining connectivity information. 

\paragraph{Exact Learnability and Edge Parity.} Our second set of results show that the identifiability barrier is not universal, but is dictated by the {\em parity of the hyperedges}.
In particular, if a hypergraph $H = (V,\cE)$ has all edges $e\in \cE$ satisfying $|e|$ is an even number (as in graphs, for instance), then $H$ can be learned using $O(\sum_{e\in \cE} 2^{|e|}\log n)$
many $\CUT$ queries. This also holds if the hypergraph is weighted with non-negative weights; see~\Cref{rem:19}.
Our algorithm proceeds by learning edge-by-edge, and therefore one can port the ideas of Nagamochi-Ibaraki~\cite{NagamochiI92} to obtain sparse $k$-certificates -- sub-hypergraphs which preserve cut-sizes up to $k$. 

\begin{rslt}\label{rslt:bnd-even}
	Let $H=(V,\cE)$ be a hypergraph on $n$ vertices and $m$ hyperedges, all of bounded size $|e|\leq r$ and even cardinality.
	Then, we can (i) reconstruct them in $O_r(m\log n)$ many $\CUT$ queries, (ii) construct a sparse $k$-certificate in $O_r(kn\log n)$ many $\CUT$-queries.
	As a corollary, we can determine whether the global or $(s,t)$ minimum cut of a bounded, even hypergraph is at most a constant $k$ or larger, in $O(n\log n)$ queries.
\end{rslt}

\noindent
We obtain our result by considering the M\"obius transform of the $\CUT$-function and using that to perform a binary-search style algorithm to learn new vertices of every edge. Unfortunately, the idea breaks down with odd-cardinality edges; in particular, after learning all but one vertex of an odd-edge, our algorithm is ambiguous about whether the current subset is the edge or whether it has one other vertex. However, this odd-vs-even problem is inherent to any algorithm; in~\Cref{sec:app?}, we show that for any odd number $k$, there exist distinct $k$-uniform hypergraphs which have the same $\CUT$-profile and therefore cannot be learned using $\CUT$ queries alone. Nevertheless, if we augmented our query model with an 
an $\EDGE$ oracle that it returns YES if a given vertex set $X$ is a hyperedge and NO otherwise, then we can resolve the above ambiguity and reconstruct {\em any} bounded hypergraph with an extra $O(mn)$-many $\EDGE$ queries.\smallskip

We also consider the case of {\em linear} hypergraphs where every pair of vertices are together present in at most one hyperedge. This is a natural sub-class of hypergraphs which generalize the notion of a ``simple undirected graph'', and have been studied extensively in extremal combinatorics (eg, see~\cite{GaoC22,JanzerST24,ImL25}).

\begin{rslt}\label{rslt:lin}
	Let $H = (V, \cE)$ be a linear hypergraph on $n$ vertices. If all $|e|\geq 4$, then $H$ can be learned in $O(\sum_{e\in \calE} |e|\log n) = O(n^2\log n)$
	many $\CUT$-queries. Irrespective of the size of $|e|$, a
	sparse $k$-certificate can be constructed using $O(\min(n^2, kn^{1.5})\log n)$ $\CUT$-queries. 
		As a corollary, we can determine whether the global or $(s,t)$ minimum cut of a linear hypergraph is at most a constant $k$ or larger, within $O(n^{1.5}\log n)$ queries.
\end{rslt}
Linearity allows us to resolve the aforementioned ambiguities in hyperedge reconstruction. For instance, if both $\{x,y,z\}$ and $\{x,y,z'\}$ appear in cuts, linearity ensures that they are part of the same hyperedge. This observation helps us bypass the need for even cardinality or bounded rank.
There is one caveat of $3$-hyperedges: there exist two distinct linear hypergraphs (see~\Cref{sec:app?}) 
	with $3$-hyperdges which have the same $\CUT$-query profile.
Fortunately, for the purpose of sparse $k$-certificates, we can replace 3-hyperedges with cut-equivalent sets of 2-edges, thus preserving correctness for our goal.

\subsection{Related Works}
Starting with the work of Rubinstein~\etal~\cite{RubinsteinSW18}, there has been significant interest in $\CUT$-query models for graphs. They described a randomized $\Ot(n)$-query algorithm for finding global minimum cuts and an $\Ot(n^{5/3})$-query algorithm for $s,t$-minimum cuts in {\em unweighted} graphs. Mukhopadhyay and Nanongkai~\cite{MukhopadhyayN20} later generalized Karger's near-linear time algorithm to provide a randomized $\Ot(n)$-query algorithm for global minimum cuts in {\em weighted} graphs. Apers~\etal~\cite{ApersEGLMN22} subsequently removed the $\log n$ factor for the unweighted case using random contraction ideas, though it remains open whether this factor can be removed for weighted graphs. 

For $s,t$-minimum cuts, the complexity has been iteratively improved from~\cite{RubinsteinSW18} to a deterministic $\Ot(n^{5/3})$ bound in unweighted graphs~\cite{AnandSW25}, and a randomized $\Ot(n^{8/5})$ bound~\cite{JiangNS26}. Concurrently, Kenneth-Mordoch and Krauthgamer~\cite{KennethK26a} study the all-pairs $s,t$-minimum cut complexity giving a $\Ot(n^{7/4})$-query algorithm, which they recently improved~\cite{KennethK26b} to $\Ot(n^{3/2})$. However, these methods fail in weighted graphs, where solving the $s,t$-minimum cut problem in $O(n^{2-c})$ queries for any $c > 0$ remains an outstanding challenge.

A primary motivation for these problems is understanding the complexity of symmetric submodular function minimization (SFM). Queyranne~\cite{Queyranne98} initially designed an $O(n^3)$ combinatorial algorithm for finding non-trivial minimizers. Using the ``isolating cuts" framework~\cite{LiP20}, Chekuri and Quanrud~\cite{ChekuriQ21} later reduced symmetric SFM to standard SFM. Combined with Jiang's $O(n^2 \log n)$ SFM algorithm~\cite{Jiang22}, this implies an $\Ot(n^2)$ bound for the symmetric case. Regarding lower bounds, while deterministic SFM requires $\Omega(n \log n)$ queries~\cite{ChakrabartyGJS22}, no $\omega(n)$ lower bound is currently known for symmetric SFM.

Recent literature has also focused on hypergraph sparsification~\cite{KoganK15, ChenKN20, KennethK23,Quanrud24,KhannaPS24}. While standard models assume a list of edges, Chen~\etal~\cite{ChenKN21} studied sparsification via $\CUT$ queries and established impossibility results due to the identifiability barriers we address. Our findings regarding even-parity hypergraphs suggest these questions warrant revisiting. Another aspect of study~\cite{AssadiCK21,KennethK25} has been the ``rounds-of-adaptivity" in these algorithms—analogous to parallel depth of algorithms. While SFM requires $\widetilde{\Omega}(n)$ rounds~\cite{ChakrabartyGJS22}, our $O(n)$-query connectivity algorithm can be implemented in polylogarithmic rounds, though the specific rounds-versus-query trade-off for hypergraph cuts is not well understood.

Finally, hypergraphs have been studied via $\ADDITIVE$ queries, which return the number of hyperedges within a subset. This model allows learning $r$-bounded hypergraphs in $O(r \log n)$ queries per edge~\cite{AngluinC05}, a task where $\CUT$ queries fail. While $\ADDITIVE$ queries can simulate $\CUT$ queries, the reverse is impossible; we suspect an exponential lower bound for such a simulation even allowing $\EDGE$ queries.

\section{Algorithm for Finding Connected Components}\label{sec:O(n)}

\noindent
Given a hypergraph $H = (V,\calE)$, we say two vertices $u$ and $v$ are connected if there is a sequence $(u = x_0, e_0, x_1, e_1, \cdots, e_t, x_t = v)$ such that 
each $x_i\in V$ and each $e_i \in \calE$, $x_0\in e_0$, and $\{x_i, x_{i+1}\}\subseteq e_i$ for $0\leq i< t$. We call such a sequence a (hyper-)walk.
Connectivity is easily seen to be a symmetric and transitive relationship.
The equivalence classes of $V$ are the connected components of $H$. 
In this section we describe a zero-error randomized algorithm which returns the connected components of an unknown hypergraph $H = (V,\cE)$ making $O(n)$-many $\CUT$ queries in expectation. 
Given the lower bound~\cite{AuzaL21} of $\Omega(n)$ for zero-error algorithms even for graphs, this resolves the complexity of this question up to constant factors.

We start with a few definitions.
Given a subset $A \subseteq V$ of vertices, we define $H[A]$, the induced subhypergraph, as $(A, \calE[A])$ where $\calE[A] := \{\emptyset \neq e\cap A~:~e\in \calE\}$. 
If $e$ has a weight $w(e)$, then $e\cap A$ inherits the same weight; $H[A]$ may have parallel copies of the same edge. We call $A$ a {\bf supervertex} if every 
pair of vertices in $A$ is connected in $H[A]$. 

We say a hyperedge $e$ {\bf touches} a supervertex $A$ if $e\cap A \neq \emptyset$. A family $\cI$ of supervertices is an {\bf independent family} 
if no hyperedge $e\in \calE$ touches two or more $A\in \cI$. We let $V(\calI) := \bigcup_{A\in \calI} A$ be the vertices of $V$ involved in $\cI$. 
Observe that $\calI$ is precisely the connected components of $H[V(\calI)]$. To see this, first note that any $A\in \cI$ is connected by definition. 
Secondly, if there is a walk from $a\in A$ to $b\in B$ in $H[V(\cI)]$, then there must exist an edge $e\in \calE[V(\cI)]$ of this walk 
which intersects two different $A'$ and $B'$ in $\cI$, contradicting the definition of independence.
Note, though, that $a$ and $b$ may be connected in $H$, however, such a path would have to be via vertices in $V\setminus V(\cI)$. \smallskip

\noindent
We can now describe our algorithm in a nutshell. We maintain a {\em collection} $\frakI := (\calI_1, \calI_2, \ldots, \calI_t)$ of independent families. 
Initially, $t = n$, and each $\calI_j$ contains a single, trivial supervertex $\{v_j\}$. Our algorithm proceeds in iterations. In each iteration, we pick two independent families 
$\cI_1$ and $\cI_2$ of roughly equal cardinality, and then {\bf merge} them to obtain a new independent set $\cI_3$ with $V(\cI_3) = V(\cI_1)\cup V(\cI_2)$. 
We delete $\cI_1$ and $\cI_2$ from $\fI$ and 
add $\cI_3$ to it. The merge step needs to find the connected components of $H[V(\calI_1) \cup V(\calI_2)]$ and is the non-trivial portion of the algorithm which we discuss subsequently. 
The algorithm terminates when $\fI$ has a single independent set $\cI$ with $V(\cI) = V$, and by design, this will be the connected components of $H$. 

\subsection{Subroutines}

We will need two subroutines from previous works, and we note these down precisely. The first is reconstructing an unknown non-negative real vector $\bx \in \mathbb{R}^n_{\ge 0}$ with $\SUM$ query access. In this model, given any subset $S\subseteq [n]$, we can get $\bx(S) := \sum_{i\in S} \bx_i$.

\begin{lemma}[zero-error and $d$-oblivious version of Theorem 2 of \cite{Choi13}]\label{lem:choi}
	 There is a randomized zero-error algorithm \Choi which, given $\SUM$-query access to an unknown $\bx \in \mathbb{R}^n_{\ge 0}$ with at most $d$ non-zero entries, reconstructs $\bx$ using $O\left(\frac{d \log n}{\log d}\right)$ queries in expectation. This algorithm doesn't need to know $d$, and runs in polynomial time.
\end{lemma}
\noindent
Choi's paper~\cite{Choi13} states a Monte-Carlo version with error probability $1/d^c$ for a large constant $c$, and this can be made into a zero-error algorithm by reverting to a $O(d\log n)$
binary-search style algorithm if we overshoot the budget. The $d$-oblivious case follows using a guess-and-check in powers-of-$2$ idea; an explicit reference is present in~\cite{ChakrabartyL23}.

The next subroutine involves reconstructing {\em spanning} forests of {\em weighted} graphs (that is, hypergraphs with all $|e|=2$) using $\CUT$ queries. 

\begin{lemma}[paraphrasing Theorem 1, \cite{LiaoC24}]\label{lem:liaoc}
	There is an adaptive, polynomial time, randomized zero-error algorithm \LiaoC, which given $\CUT$-query access to an $n$-vertex, weighted graph $G=(V,E)$,
	finds a spanning forest of $G$ making $O(n)$-many $\CUT$ queries in expectation.
\end{lemma}

\subsection{Merging Independent Families via Graph Learning}

We now describe the merge procedure which takes two independent families, $\cI = (I_1, \cdots, I_s)$ and $\cJ = (J_1, \cdots, J_t)$, and returns
$\cK = (K_1, \ldots, K_k)$, where $\cK$ are the connected components of $H[\bigcup_{i=1}^s I_i \cup \bigcup_{j=1}^t J_j]$. 

To this end, we define  the {\bf cross graph} induced by $\cI$ and $\cJ$. 
The cross graph $G := G[\cI, \cJ]$ is a {\em weighted} bipartite graph. The vertex set is $\{i_1, \ldots, i_s\} \cup \{j_1, \ldots, j_t\}$ corresponding
to the supervertices in $\calI$ and $\calJ$. Below we use $[s]$ to denote $\{1,2,\ldots, s\}$ and $[t]$ to denote $\{1,2,\ldots, t\}$.
The edge set $F$ is constructed as follows.
We go over all hyperedges $e\in E(H)$, and 
\begin{asparaitem}
	\item if $e$ touches $I_a$ and $J_b$ {\em and} $e\subseteq I_a \cup J_b$, then add $(i_a,j_b)$ to $F$ with weight $w(e)$.
	\item if $e$ touches $I_a$ and $J_b$ {\em and} $e\not\subseteq I_a \cup J_b$, then add $(i_a,j_b)$ to $F$ with weight $w(e)/2$.
\end{asparaitem}
Crucially note that an edge $e$ cannot touch $I_a$ and $I_{a'}$ with $a\neq a'$ since $\calI$ is an independent family. Similarly, $e$ cannot touch $J_b$ and $J_{b'}$
with $b\neq b'$. It could happen that $e$ touches $I_a$ but not $J_b$; in that case we do not add any new edge to $F$. Note that the constructed bipartite graph is {\em weighted} even if the hypergraph is unweighted.

For $a\in [s]$, we let $\deg_G(i_a)$ denote the {\em weight} of edges in $F$ incident on $i_a$. If $\deg_G(i_a) > 0$, we call $I_a$ a {\bf participating supervertex}.
Similarly for $b\in [t]$. We let $G_+ \subseteq G$ denote the graph induced by participating supervertices; note that $G$ and $G_+$ have the same edges $F$. \smallskip

\noindent
The main reason to use the cross graph is that it captures the connectivity information of $H[V(\cI) \cup V(\cJ)]$ as is encapsulated in the following lemma.
\begin{lemma}\label{lem:connG}
	Given $\cI$ and $\cJ$ and the corresponding cross graph $G$ defined above, let $C_1, \cdots, C_k$ be the connected components of $G$.
	Define \[	\cK := \Big\{ ~\bigcup_{i_a \in C_\ell} I_a \cup \bigcup_{j_b\in C_\ell} J_b~~~:~\ell\in [k] ~\Big\}\]
	Then $\cK$ is independent and is precisely the connected components of $H[V(\cI) \cup V(\cJ)]$.
\end{lemma}
\begin{proof}
	For brevity let $U := V(\cI) \cup V(\cJ)$. By design, $\cK$ is a partition of $U$; let us call this $(K_1, \ldots, K_k)$.
	To see that $\cK$ is independent, fix any $e\in \calE(H)$ and suppose $e$ touches $K_p$ and $K_q$. Since $\cI$ and $\cJ$ are independent, 
	$e$ must touch some $I_a$ and $J_b$, where $i_a\in C_p$ and $j_b\in C_{q}$ with $p \neq q$. But this contradicts 
	$C_p$ and $C_{q}$ were different connected components in $G$ (since the $(i_a,j_b)$ edge is present in $G$).

	All that remains is to show a single $K\in \cK$ is connected. Fix two vertices $x$ and $y$ in $K$, and let $x\in I_a$ and $y \in J_b$ ($y$ could also be in some $I_b$; the same argument would work).
	Since $i_a$ and $j_b$ are in the same $C_\ell$ of $G$, there is a path $(i_a = u_0, u_1, \cdots, u_r = j_b)$ in $G$. Each $(u_p, u_{p+1})$ edge in $G$ corresponds to an edge $e\in E$ which intersects
	two supervertices, one in $\cI$ and another in $\cJ$. Using these hyperedges, and the fact that vertices in a supervertex are connected, we can show that 
	$x$ and $y$ are connected in $H[V(\cI) \cup V(\cJ)]$. \qedhere
\end{proof}

The above lemma reduces the problem of merging to find the connected components of a weighted bipartite {\em graph}. This problem has been studied by~\cite{LiaoC24} and we will use that. 
However, we do need to ensure that we can simulate $\CUT$-queries ``cross queries'' on $G$ using cut queries on the hypergraph $H$. 
We do this next. To simplify matters, we first note that in a bipartite graph $G$ with $L\cup R$ being the bipartition, a $\CUT(S)$-query is equivalent (up to constant factors)
to a ``cross query'' defined as follows: given $A\subseteq L$ and $B\subseteq R$, $\CROSS(A,B)$ denotes the total weight of edges with one endpoint in $A$ and the other in $B$.
Given access to $\CROSS$ queries, we can evaluate cut queries since $\CUT(S) = \CROSS(S\cap L, B) + \CROSS(A, S\cap R) - \CROSS(S\cap L, S\cap R)$.
We now show how to simulate $\CROSS$ queries in the weighted bipartite graph $G$ using $\CUT$ queries on the hypergraph $H$.

\begin{lemma}\label{lem:simul}
	Let $A\subseteq \{i_1, \ldots, i_s\}$ and $B\subseteq \{j_1, \ldots, j_t\}$. Recall,  $\CROSS(A, B)$ denote the total weight of edges $(i_a, j_b)\in F$ with $i_a\in A$ and $j_b\in B$. 
	Let $P := \bigcup_{i_a\in A} I_a$ and $Q := \bigcup_{j_b\in B} J_b$.
	Then,
		\[
			\CROSS(A,B) = \frac{1}{2}(\CUT_H(P) + \CUT_H(Q) - \CUT_H(P \cup Q))
		\]
\end{lemma}
\begin{proof}
	Note $\CROSS(A,B) = \sum_{(i_a,j_b)\in E(G):i_a\in A, j_b\in B} w(i_a,j_b)$ where note $E(G)$ can have multiple parallel edges. 
	By the way $G$ was defined, we can map each $(i_a,j_b) \in E(G)$ to a {\em unique} $e\in E(H)$. 
	Indeed, this edge $e$ must touch $I_a$ and $J_b$. This is unique because $e$ cannot touch any $I_{a'}$ or $J_{b'}$ since $\cI$ and $\cJ$ are independent.
	Furthermore, if $e\subseteq I_a\cup J_b$, then $w_G(i_a,j_b) = w_H(e)$, where we add the subscripts to clarify where the weights are measured.
	Let $E_1$ be the hyperedges $e$ touching some $I_a$ and $J_b$ but not a subset of $I_a\cup J_b$, while $E_2$ be the hyperedges $e'$ touching some $I_a$ and $J_b$ 
	and is a subset of $I_a\cup J_b$.
	So we get, 
	\[
		\CROSS(A,B) = \sum_{e\in E_1} \frac{w(e)}{2} + \sum_{e\in E_2} w(e)
	\]
	Now, we consider the quantity RHS $:= \frac{1}{2}(\CUT_H(P) + \CUT_H(Q) - \CUT_H(P\cup Q))$. 
	Let $e$ be a hyperedge in $E(H)\setminus (E_1 \cup E_2)$. There are a few cases: (a) $e$ doesn't touch $P$ or $Q$, 
	(b) $e$ touches $P$, or (c) $e$ touches $Q$. Note, $e$ cannot touch both $P$ and $Q$; if it does, then it touches some $I_a$ and $J_b$ and would thus belong to $E_1 \cup E_2$.
	Check that $e$'s contribution to RHS is $0$ in all three cases; in case (a) it contributes to $0$ to all three terms, and (b), (c) it contributes $w(e)$ to exactly one of $\CUT_H(P)$ or $\CUT_H(Q)$,
	and also to $\CUT_H(P\cup Q)$.
	
	Next we consider the contributions of $e\in E_1\cup E_2$.
	First consider an edge $e\in E_2$. Since $e\subseteq I_a \cup J_b \subseteq P\cup Q$, $e$ contributes $w(e)$ to $\CUT_H(P)$ and $\CUT_H(Q)$ but not to $\CUT_H(P\cup Q)$. Thus, its contribution to RHS is $w(e)$.
	Next consider an edge $e\in E_1$. $e$ intersects $I_a$ and $J_b$ but no other set in $\cI \cup \cJ$, and yet it is not a subset of $I_a\cup J_b$. This implies that $e$ contains a vertex
	$v \in V\setminus (V(\cI) \cup V(\cJ))$. This means $e$ contributes $w(e)$ to {\em all three}: $\CUT_H(P)$, $\CUT_H(Q)$ and $\CUT_H(P\cup Q)$. So its contribution to RHS is $w(e)/2$.
	This shows that the RHS, summed over all edges $e\in E(H)$, gives $\CROSS(A,B)$.
\end{proof}
\noindent
At this point, we could apply the algorithm from~\cite{LiaoC24}, however, there is one issue. That algorithm makes queries whose number is  linear (in expectation) of the number of {\em vertices} of $G$ which is $s+t$.
However, to argue about the query complexity, it is crucial that we run the process on $G_+$ rather than $G$; recall $G_+$ was the restriction of $G$ to the participating supervertices.
In short, we don't want to ``waste queries'' on $i_a$ or $j_b$ which have $\deg_G(\cdot) = 0$. Recognizing the degree-0 vertices can be done using coin-weighing algorithms as follows. 

\begin{lemma}\label{lem:part}
	Fix $\cI$ and $\cJ$ and the graph cross-graph $G$ as above. Define vector $\bx \in \RR^s_+$ where $\bx[a] = \deg_G(i_a)$; similarly define $\by \in \RR^t_+$
	with $\by[b] = \deg_G(j_b)$. One can simulate a $\SUM$ query on $\bx$ and $\by$ using $O(1)$ many $\CUT_H$ queries on the hypergraph. 
\end{lemma}
\begin{proof}
	For a subset $A\subseteq \{i_1, \ldots, i_s\}$, the $\SUM(\bx(A))$ is precisely $\CROSS(A,\{j_1, \ldots, j_t\})$. The other direction is symmetric.
\end{proof}
\noindent
Let $d_\bx$ and $d_\by$ be the {\em support} of $\bx$ and $\by$; in other words, it is the number of participating supervertices in $\cI$ and $\cJ$ respectively.
The following is another simple observation.
\begin{lemma}\label{lem:obs}
	Let $\cK$ be as defined in \Cref{lem:connG} and let $d := |\calI| + |\calJ| - |\calK|$. Then $d\leq d_\bx + d_\by \leq 2d$.
\end{lemma}
\begin{proof}
	If $(C_1, \ldots, C_k)$ are the connected components of $G$, then the lemma 
	follows from noting that $d = (s + t - k) = \sum_{\ell=1}^k \left(|C_\ell| - 1\right)$ and $d_\bx + d_\by = \sum_{\ell: |C_\ell| \geq 2} |C_\ell|$.
	Here $|C_\ell|$ denotes the number of supervertices in the components and the participating supervertices are the ones which are in components of size $\geq 2$ since they have an edge incident on them.
\end{proof}
\noindent
Now we have all the ingredients to specify the {\sc Merge} procedure. 

\begin{algorithm}[ht!]
	\caption{Merging Independent Families}\label{alg:merge}
	%
	\begin{varwidth}{\dimexpr\linewidth-2\fboxsep-2\fboxrule\relax}
		\begin{algorithmic}[1]
			\Procedure{Merge}{$\cI,\cJ$}:
			\LineComment{Input: Two independent families $\cI, \cJ$}
			\LineComment{Output: $\cK$ as described in~\Cref{lem:connG}}
				\State Simulate $\CROSS$ query on $G(\cI,\cJ)$ as described in \Cref{lem:simul}
				\State Learn vectors $\bx$ and $\by$ as described in~\Cref{lem:part} using \Choi (\Cref{lem:choi})  \label{alg:learn-vec}
				\State Using $\bx, \by$ and \LiaoC (\Cref{lem:liaoc})  learn spanning forest of $G_+$. \label{alg:learn-spf}
				\State Using connected components of $G_+$ return $\cK$ a la~\Cref{lem:connG}
			\EndProcedure
		\end{algorithmic}
	\end{varwidth}%
\end{algorithm}
\begin{lemma}\label{lem:qc-merge}
	Let $\cI$ and $\cJ$ be two independent families, and let $\cK$ be the output of {\sc Merge}($\cI, \cJ)$.
	Let $d:= |\cI| + |\cJ| - |\cK|$. {\sc Merge}, in expectation, makes $O\left(\frac{d\log \max(|\cI|, |\cJ|)}{\log d}\right)$
	many $\CUT$ queries.
\end{lemma}
\begin{proof}
	The query complexity is dominated by \Choi~ which makes (by \Cref{lem:choi}) \\
	$O\left(\frac{d_\bx \log s}{\log d_\bx} + \frac{d_\by \log t}{\log d_\by}\right)$-many queries (using~\Cref{lem:part}).
	This is in budget using~\Cref{lem:obs}. \LiaoC~is randomized but the expected number of queries (\Cref{lem:liaoc}) is $O(d_\bx + d_\by) = O(d)$ by~\Cref{lem:obs}; since $d \leq \max(s, t)$, 
	this is also within budget.
\end{proof}

\subsection{Connected Components in Hypergraphs}
Given the {\sc Merge} subroutine, we can complete the description of the full algorithm. 
As mentioned earlier, 
we always maintain a {\em collection} $\frakI := (\calI_1, \calI_2, \ldots, \calI_t)$ of independent families. We begin with $t=n$ and
initially  $\calI_j$ contains a single, trivial supervertex $\{v_j\}$. The algorithm proceeds in iterations, and in each iteration, we pick two independent families of nearly equal size and merge them.
When we can't do so any more, we have independent families of differing sizes, and so there can be only $O(\log n)$ many families. In this case, we can just sequentially merge them again. 
The algorithm pseudocode is below.

\begin{algorithm}[ht!]
	\caption{Find Connected Components }\label{alg:main}
	%
	\begin{varwidth}{\dimexpr\linewidth-2\fboxsep-2\fboxrule\relax}
		\begin{algorithmic}[1]
			\Procedure{FindCC}{$V,\CUT$}:
			\LineComment{Input: $n$ vertices $V$ and $\CUT$ query access to hidden hypergraph $H$}
			\LineComment{Output: connected components of $H$}			
			\State Create $\fI \eq \{\cI_1, \ldots, \cI_n\}$ where $\cI_i := \{\{v_i\}\}$
			\While{$\exists \cI, \cJ \in \fI~:~ |\cI|/|\cJ| \in [1/2,2]$}:
				\State $\cK \eq ${\sc Merge}$(\cI, \cJ)$. \label{alg:q1}
				\State $\fI \eq \fI - \{\cI,\cJ\} + \cK$; 
			\EndWhile
			\LineComment{At this point there can be at most $\ceil{\log_2 n}$ elements in $\fI$}
			\LineComment{Merge all these families in any order}
			\State Let $\fI = \{\cI_1, \ldots, \cI_\ell\}$ with $\ell \leq \ceil{\log_2 n}$; $\cI \eq \cI_1$
			\For{$2\leq t\leq \ell$}:
				\State $\cI \eq $ {\sc Merge}($\cI_t,\cI$); $\fI \eq \fI - \cI_t$ \label{alg:q2}
			\EndFor
			\State \Return $\cI$
			\EndProcedure
		\end{algorithmic}
	\end{varwidth}%
\end{algorithm}
\noindent
The correctness of the algorithm follows from the correctness of {\sc Merge}. We now show that the total number of $\CUT$ queries made is $O(n)$ in expectation. 
Note that all queries are made in the {\sc Merge} calls in Line~\ref{alg:q1} and Line~\ref{alg:q2}. We treat these separately. \smallskip

\noindent
{\bf Queries made in Line~\ref{alg:q1}.} Recall that we use $d:= |\cI| + |\cJ| - |\cK|$; without loss of generality, let us assume $|\cI| \geq |\cJ|$; note $|\cI| \le 2|\cJ|$.
We call a $\cK \eq $ {\sc Merge}($\cI, \cJ)$ operation {\em thick} if $d \geq \sqrt{|\cI|}$. Otherwise, we call a {\sc Merge} operation {\em thin}.  

By~\Cref{lem:qc-merge}, the expected number of queries made by a thick {\sc Merge}$(\cI, \cJ)$ is $O\left(\frac{d\log |\cI|}{\log d}\right)$ which is $O(d)$ since the denominator cancels
the logarithmic term in the numerator. Furthermore, consider the potential $\Phi(\fI) := \sum_{\cI \in \fI} |\cI|$; a {\sc Merge} drops $\Phi(\fI)$ by exactly $d$.
So {\em all} thick merges can be charged to, up to a multiplicative constant factor, to the total drop of $\Phi(\fI)$. In the beginning $\Phi = n$, and so, the total query complexity, in expectation, 
of all thick merges is $O(n)$. 

To address thin merges, we call an independent family $\cI$ to be in class $t$ if $|\cI|  \in [2^t, 2^{t+1})$, 
and classify a thin {\sc Merge}($\cI, \cJ$) to be in class $t$ if $\min(|\cI|, |\cJ|)$ is a class $t$ family.
Note that after such a thin merge, we replace $\cI$ and $\cJ$ by $\cK$ with $|\cK| \geq 2^{t+1} - 2^{t/2}$.
Thus $\cK$ is almost not in class $t$. More precisely, let's say a vertex $v \in V$ {\em participates} in a class $t$ merge
if it is present in some supervertex of  $\cI$ or $\cJ$. What the above calculation shows is that the same vertex {\em cannot} participate in three 
class $t$-thin merges; the second class $t$ merge would put $v$ in an independent set which is in class $t+1$ or higher.
As a result, we can assert that the number of class $t$-thin merges is $\leq \frac{2n}{2^t}$.

By~\Cref{lem:qc-merge}, a class $t$-thin merge takes, in expectation, at most $O\left(\frac{d\log (2^{t+1})}{\log d}\right)$ many $\CUT$ queries.
Since $\frac{d}{\log d}$ is an increasing function of $d$, and $d = O(2^{t/2})$ since the merge is thin, the number of expected queries 
is at most $O\left(\frac{2^{t/2}\log (2^{t+1})}{\log 2^{t/2}}\right) = O(2^{t/2})$. Combining with the upper bound on the number of class $t$-thin merges, we get
that the total expected number of CUT queries made over thin merges is at most
	$\sum_{t=0}^{\floor{\log_2 n}} \frac{2n}{2^t} \cdot O(2^{t/2}) = O(n)$. \smallskip
	
\noindent
{\bf Queries made in Line~\ref{alg:q2}.} We leverage the fact that the {\em number} of {\sc Merge}s is $O(\log n)$. 
In particular, let $d_t$ denote the drop $|\cI| + |\cI_t| - |\cI_{\mathrm{new}}|$ at the $t$th for-loop.
Again by~\Cref{lem:qc-merge}, the expected number of queries is at most $O\left(\frac{d_t \log n}{\log d_t}\right)$, and so 
the total number of queries made in~\Cref{alg:q2} is $\sum_{t=1}^\ell \frac{d_t \log n}{\log d_t}$. 
Now, when $\ell = O(\log n)$, one can show that $\sum_{t=1}^\ell \frac{d_t}{\log d_t} \leq O(n/\log n)$. 
Indeed, if $d_t \leq \frac{n}{\log^2 n}$, then we use the number of $t$'s is $O(\log n)$ to bound their contribution to the sum by $O(n/\log n)$.
If $d_t > \frac{n}{\log^2 n}$, then $\frac{d_t}{\log d_t} = O(d_t/\log n)$ and so their contribution is at most $O(\frac{\sum_t d_t}{\log n})$. 
Once again $\sum_t d_t$ can be charged to $\Phi(\fI)$ initially which is $n$. 
All in all, we get the following theorem which establishes~\Cref{rslt:O(n)}.
\begin{theorem}
	\Cref{alg:main} returns the connected components of a hypergraph in an expected $O(n)$ many $\CUT$ queries.
\end{theorem} 

\section{Learning \& Sparse Connectivity Certificates for Classes of Hypergraphs}

As noted in the introduction, in general a hypergraph cannot be learned from $\CUT$ queries since there are distinct hypergraphs with the same cut-profile. 
In this section, we show that this impossibility result is rather delicate and depends on the {\em parity} of the sizes of the hyperedges!
In particular, if all the hyperedges are even, then we can learn the hypergraph using $\CUT$ queries; there is no ambiguity.

\begin{theorem}\label{thm:bnd-even-learn}
	If $H = (V, \cE)$ is a hypergraph with $|e|$ an {\em even} number for all $e\in \calE$, then we can efficiently learn $H$ in $O(\sum_{e\in \calE} 2^{|e|}\cdot \log n)$ many $\CUT$-queries.
	If $H$ was furthermore $r$-bounded, that is, $|e|\leq r$ for all $e\in \calE$, then the query complexity is $O(2^rm\log n)$ where $m$ is the number of hyperedges.
\end{theorem}
\noindent
Indeed, we can say more underscoring the delicateness of the ``reconstruction impossibility'' even more. 
Suppose we also had access to a Boolean oracle $\EDGE(S)$ which said True if $S\in \cE$ and False otherwise. 
This is similar to ``pair'' queries in a normal graph where given a pair one asks if it is an edge or not. 
Clearly, one can reconstruct an $r$-bounded hypergraph $H$ using 
$O(n^r)$ many $\EDGE$-queries. We can do {\em much better} when
$\CUT$ {\em and} $\EDGE$ queries are given, and the hypergraph is sparse.

\begin{theorem}\label{thm:xtra-cut+edge-learn}
	Any hypergraph $H = (V,\cE)$ with $|V| = n$ and $|\cE| = m$ and $|e|\leq r$ for all $e\in \cE$
	can be reconstructed using $O(2^r m\log n)$-many $\CUT$ queries and $O(mn)$-many $\EDGE$-queries.
\end{theorem}

\noindent
The second class of hypergraphs which can be learned only using $\CUT$ queries are {\em linear} hypergraphs. A hypergraph $H = (V,\calE)$ is called linear if $|e\cap f|\leq 1$ for any distinct
$e,f\in \calE$. Note that this generalizes the notion of ``{\em simple}'' undirected graphs which states that between any two vertices, there can be at most one (hyper) edge.
In general linear hypergraphs cannot be learned (see~\Cref{sec:app?}), but if $|e|\geq 4$ for all $e\in \calE$, then we can learn it (even if it has edges with odd number of vertices).

\begin{theorem}\label{thm:linear-learn}
	If $H = (V, \cE)$ is a {\em linear} hypergraph with $|e| \geq 4$, then we can 
	learn $H$ in $O(\sum_{e\in \calE} |e|\log n) \leq O(n^2\log n)$ many $\CUT$-queries.
\end{theorem}

\paragraph{Spanning Subhypergraphs and Sparse Certificates.} 
Given a hypergraph $H = (V,\cE)$ with potentially parallel copies of a hyperedge, and a natural number $k$, a {\em sparse $k$-certificate} (as in~\cite{NagamochiI92},~\cite{ChekuriX17},~\cite{GuhaMT15}) is a subhypergraph 
$H' = (V, \cE')$ with $\cE' \subseteq \cE$
such that for any subset $A\subseteq V$, we have $|\partial_{H'}(A)| \geq \min(k, |\partial_{H}(A)|)$. 

In particular, $H'$ preserves all cuts up to size $k$, and in particular, if the minimum cut of $H$ is $\leq k$, then so is the minimum cut in $H'$.
A $1$-sparse certificate is a {\em spanning} subhypergraph. The following result, essentially due to Nagamochi and Ibaraki~\cite{NagamochiI92} 
shows that if we have the ability to find {\em spanning} subhypergraphs, 
then we can construct $k$-sparse certificates by repeatedly finding and peeling.

\begin{lemma}[Greedy Peeling a la Nagamochi-Ibaraki~\cite{NagamochiI92}]\label{thm:ni}
	Given a hypergraph $H = (V,\cE)$ on $n$ vertices with potentially parallel copies of hyperedges, if we can find a spanning subhypergraph $F \subseteq H$ in $q(n)$ many $\CUT$-queries, 
	then we can construct a $k$-sparse certificate in $O(kq(n))$ many $\CUT$-queries.
\end{lemma}
\begin{proof}
	This is a greedy peeling argument. If $H_1, H_2, \ldots, H_k$ are subhypergraphs satisfying the property 
	that $H_i$ is a spanning hypergraph of $H - \cup_{j=1}^{i-1} H_j$, then~\cite{NagamochiI92,GuhaMT15,ChekuriX17} show that 
	$\cup_{i=1}^k H_i$ is a $k$-sparse certificate. We refer the reader to these papers for this result.
\end{proof}

All our algorithms described above to learn hypergraphs
proceed by learning ``edge-by-edge'', and in particular, given a vertex $x$ and a subset $U\subseteq V\setminus x$, it can learn an edge incident to both $x$ and $U$. Using standard Prim-style algorithms, one can then learn a spanning subhypergraph. Using the above theorem, we can get sparse $k$-certificates.
In particular, they give the following results.

\begin{theorem}\label{thm:bnd-even-sparse}
	If $H = (V, \cE)$ is a hypergraph with $|e|\leq r$ an {\em even} number for all $e\in \calE$, then we can efficiently construct 
	a sparse $k$-certificate $H$ in $O(k2^rn\log n)$ many $\CUT$-queries. 
\end{theorem}
For linear hypergraphs, we can bypass the bounded hyperedge size condition when we construct sparse $k$-sparsifiers. This gives the following subquadratic guarantee 
for constant $k$.
\begin{theorem}\label{thm:linear-learn-sparse}
	If $H=(V,\cE)$ is a {\em linear} hypergraph, then we can find a
	cut-equivalent sparse $k$-certificate, not necessarily a subhypergraph of $H$,
	in $O(kn^{1.5}\log n)$ many $\CUT$-queries.\end{theorem}

\subsection{Tools: M\"obius Transform and Binary Search}

Our algorithm to learn a hyperedge proceeds by learning it one vertex at a time. In particular, if $e$ is an edge, and suppose we know a subset $X\subseteq e$ of vertices of $e$.
To learn another vertex of $e$, we would like a tool which answers a question of the following form: given a subset $X\subseteq V$, does there 
exist a hyperedge $e$ which contains {\em all} vertices of $X$ and something outside? If we have such a tool, we can use a binary-search style technique to find the next vertex outside 
in $O(\log n)$ calls. On the other hand, $\CUT(X)$ counts the number of hyperedges which intersects {\em some} vertex of $X$ and something outside. To get the former from $\CUT$ queries, 
the idea is to look at the M\"obius transform\footnote{The M\"obius transform of a set function $f$ is $g(X) := \sum_{S\subseteq X} (-1)^{|X\setminus S|}f(X)$} of the $\CUT$-function.
Unfortunately, this tool doesn't work all the time (as the impossibility results show), but when the parity is favorable, this idea goes through. 

We begin with a few definitions.
\begin{definition}\label{def:ghA}
	Given a subset $X\subseteq V$ and a non-empty $A\subseteq V\setminus X$, define the functions \\$g(X; A) := \abs{\{e\in \calE~:~ X\subseteq e, e\cap A \neq \emptyset, ~\textrm{and}~~e\cap V\setminus (X\cup A) \neq \emptyset\}}$, and \\
	 $h(X; A) := \abs{\{e\in \calE~:~ X\subsetneq e, \textrm{and}~~e\subseteq A \cup X\}}$.
\end{definition}
\noindent
In plain English, $g(X;A)$ counts the number of hyperedges that contain all vertices of $X$, intersects both $A$ and $V\setminus (A\cup X)$, while 
$h(X;A)$ counts the hyperedges that contain all of $X$, intersects $A$,  and is contained in $A\cup X$ completely. So, $g(X;A) + h(X;A)$ counts the number of hyperedges
that contain all of $X$ and touches some vertex of $A$. Note that when $A = V\setminus X$, this is precisely the tool described in the first paragraph.
The next lemma shows that one can obtain some information depending on the parity of $|X|$, and subsequently we show its utility.

\begin{remark}\label{rem:19}
	Although we restrict our attention to {\em unweighted} hypergraphs, our reconstruction goes through for weighted hypergraphs 
	with non-negative weights. For this, we need to modify the definition of $g(X;A)$ and $h(X;A)$ to be the total weight of the edges
	in the subset of their RHS instead of cardinality. 
\end{remark}

\begin{lemma}[M\"obius Transform]\label{lem:mobA}
	For a subset $X\subseteq V$ and non-empty $A\subseteq V\setminus X$, define 
	\begin{equation}\label{eq:qXA}	
		Q(X; A) := \CUT(A) + \sum_{\emptyset \neq S\subseteq X} (-1)^{|S|+1}\left(\CUT(S) - \CUT(S\cup A)\right) 
	\end{equation}
	Then, we have
	\[
	Q(X;A) = \begin{cases}
		g(X; A) & \textrm{if $|X|$ even} \\
		g(X; A) + 2h(X;A) & \textrm{if $|X|$ odd}
	\end{cases}
	\]
		Note that $Q(X;A)$ can be computed in $2^{|X|+1}$ many $\CUT$-queries.
\end{lemma}
\noindent
We defer the proof of the above to~\Cref{app:gen-mt} where we prove a generalization to subpartitions. We next show the utility of these queries.

Before that, we make a few observations.

\begin{observation}\label{obs:XA}
Given disjoint subsets $X$ and $A$,  if $Q(X;A) > 0$, then this implies there exists an edge $e\in \calE$ such that $X\subseteq e$ and $e\cap A \neq \emptyset$ (irrespective of the parity of $X$).
In particular, if $A = \{a\}$, then $X\cup \{a\}\subseteq e$.
Furthermore, if $|X|$ is odd, then it is an if-and-only-if statement. More precisely, if $|X|$ is odd and $Q(X;A) = 0$, then both $g(X;A)$ and $h(X;A)$ are $0$
which means there is no edge $e$ of the form $X\subseteq e$ and $e\cap A \neq \emptyset$.	
\end{observation}

\begin{observation}\label{obs:sum}
	Given disjoint subsets $X$ and $A$ and any partition $(A_1, A_2)$ of $A$, 
	we have $Q(X;A) \leq Q(X;A_1) + Q(X;A_2)$
\end{observation}
\begin{proof}
	One could obtain this via submodularity of $\CUT$ in the definition \eqref{eq:qXA} of $Q(X;A)$, or one just uses~\Cref{lem:mobA}
	and the simpler-to-verify statements: $g(X;A_1) + g(X;A_2) \geq g(X;A)$ and $h(X;A_1) + h(X;A_2) \geq h(X;A)$.
	Indeed, any edge $e\in \cE$ counted in the RHS-es is either counted in the LHS or RHS, and any $e$ intersecting both $A_1$ and $A_2$ are
	double counted in the LHS-es.
\end{proof}
\noindent
Before we state our binary-search subroutine, we need one last standard fact. 
\begin{claim}\label{clm:detbin}
	Given any set $A$ with $|A| \geq 2$, there exists explicit partitions $(A^{(\ell)}_1, A^{(\ell)}_2)$ of $A$ for $1\leq \ell \leq \ceil{\log_2 |A|}$
	with the property that for any distinct $x,y \in A$, there is some $\ell$ such that $|A^{(\ell)}_1 \cap \{x,y\}| = |A^{(\ell)}_2 \cap \{x,y\}| = 1$.
	That is, $x$ and $y$ are split asunder in at least one of the partitions.
\end{claim}
\begin{proof}
	Index the elements of $|A|$ as $\ceil{\log_2|A|}$-dimensional bit-vectors, and define $A^{(\ell)}_1 := \{x\in A: x_\ell = 1\}$ and $A^{(\ell)}_2 = A\setminus A^{(\ell)}_1$.
	Any two $x$ and $y$ differ in at least one bit $\ell$ and that suffices.
\end{proof}
\noindent
Now we are ready to describe the binary search subroutine that uses the $Q(X;A)$'s.

\begin{algorithm}[ht!]
	\caption{Binary Search}\label{alg:bins}
	\begin{varwidth}{\dimexpr\linewidth-2\fboxsep-2\fboxrule\relax}
		\begin{algorithmic}[1]
			\Procedure{BinSearch}{$X,A$}:
			\LineComment{Input: non-empty subsets $X \subseteq V$, $A \subseteq V\setminus X$}
			\LineComment{Output: $a\in A$ or $\bot$ satisfying conditions~\Cref{lem:binsearch-edge-growing}}			
			
			\If{$|X|$ odd}:
				\If{$Q(X;A) > 0$}:\Comment{Binary Search in $A$}
					\While{$|A| > 1$}:\label{line:wloop-start}
						\State $(A_1, A_2)$ arbitrary equipartition of $A$
						\State ~{\bf if} $Q(X;A_1) > 0$ {\bf then}: $A\eq A_1$ {\bf else} $A\eq A_2$
					\EndWhile
					\State \Return singleton vertex of $A$ \label{line:wloop-end}
				\Else: \Comment{$Q(X;A) = 0$}
					\State \Return $\bot$
				\EndIf
			\Else: \Comment{$|X|$ even}
				\If{$|A| = 1$}:
					\State \Return $\bot$
				\Else:
					\State Construct $(A^{(\ell)}_1, A^{(\ell)}_2)$ for $1\leq \ell \leq \ceil{\log_2|A|}$ as described in~\Cref{clm:detbin}
					\State Find $\ell$ such that $Q(X;A^{(\ell)}_1)$ and $Q(X;A^{(\ell)}_2)$ are both $> 0$; if none, return $\bot$. \label{line:bot2}
					\State Let $A\eq A^{(\ell)}_1$ \Comment{Note: $Q(X;A^{(\ell)}_1) > 0$}
					\State Run Lines~\ref{line:wloop-start} to~\ref{line:wloop-end}.
				\EndIf
			\EndIf
			\EndProcedure
		\end{algorithmic}
	\end{varwidth}%
\end{algorithm}

\begin{lemma}[Spec of {\sc BinSearch}($X,A$)]\label{lem:binsearch-edge-growing}
	Let $X\subseteq V$ and non-empty $A\subseteq V\setminus X$ be two arbitrary subsets. 
	Then,
	\begin{asparaenum}
		\item  when $|X|$ is odd, {\sc BinSearch}($X,A$) either returns an $a\in A$ such that there exists an edge $e \in \cE$ with $X+a \subseteq e$,
		or it returns $\bot$ and
		we can assert there is {\em no} edge $e\in \cE$ such that $X\subseteq e$ and $|e\cap A| \geq 1$. 
		\item  when $|X|$ is even, {\sc BinSearch}($X,A$) either returns an $a\in A$ such that there exists an edge $e \in \cE$ with $X+a \subseteq e$,
		or it returns $\bot$ and
		we can assert there is {\em no} edge $e\in \cE$ such that $X\subseteq e$ and $|e\cap A| \geq 2$. 
	\end{asparaenum}
	
	The algorithm makes $O(\log n)$ many queries to $Q(X;\cdot)$ and so can be simulated using $O\left(2^{|X|+1}\log n\right)$ many $\CUT$-queries.
\end{lemma}
\noindent

\begin{proof}
	Let us consider the case when $|X|$ is odd. 
	We first query $Q(X;A)$. If $Q(X;A) = 0$, then by~\Cref{obs:XA} we know there is no edge $e\in \cE$ with $X\subseteq e$ and $e\cap A \neq \emptyset$, and thus we return $\bot$.
	Otherwise,  we split $A$ arbitrarily into $(A_1, A_2)$ such that each is of size $\leq \ceil{|A|/2}$.
	By~\Cref{obs:sum}, we get $Q(X;A) \leq Q(X;A_1) + Q(X;A_2)$, and so one of $Q(X;A_1)$ or $Q(X;A_2)$ is positive, and we recurse on that half. 
	We continue so till we obtain a singleton vertex $a \in A$ such that 
	$Q(X;\{a\}) > 0$. This means there is at least one edge $e\in \cE$ such that $X\subseteq e$ and $e\cap \{a\} \neq \emptyset$. In other words, $X+a\subseteq e$. 
	\smallskip
	
	\indent
	Now let us consider the more interesting case when $|X|$ is even. In this case $Q(X;A) = 0$ {\em doesn't} mean the non-existence of an edge ``extending'' $X$.
	 However, if there exists an edge $e\in \cE$ such that $X\subseteq e$ and $|e\cap A| \geq 2$.
	Then using~\Cref{clm:detbin}, there is an $1\leq \ell \leq \ceil{\log_2|A|}$ such that $e\cap A^{(\ell)}_1 \neq \emptyset$ and $e\cap A^{(\ell)}_2\neq \emptyset$.
	That is, $g(X;A^{(\ell)}_i) > 0$ for $i\in \{1,2\}$. Since $|X|$ is even, we have $Q(X;A^{(\ell)}_i) > 0$ for $i\in \{1,2\}$.  
	Contrapositively, if this isn't true, then we can assert the non-existence of any edge $e\in \cE$ such that $X\subseteq e$ and $|e\cap A| \geq 2$.
	This is what we do in Line~\ref{line:bot2} by returning $\bot$.
	Once we get our hand on such an $\ell$, we then proceed as in the previous case to obtain an 
	$a\in A^{(\ell)}_1$ such that $Q(X;\{a\})>0$, which, as in the previous case, implies the existence of the edge $e$ with $X+a\subseteq e$.
\end{proof}
The last preliminary observation is that once we find an edge $e\in H$, we can simulate $\CUT$-queries on the hypergraph $H' := H - e$.
This is simply noting that $\CUT_{H'}(S) = \CUT_H(S) - 1$ if $e\cap S$ and $e\cap V\setminus S$ are both non-empty, and equals to $\CUT_H(S)$ otherwise.
We can repeat this for any subset of edges.
\begin{observation}\label{obs:triv}
	If $F\subseteq \calE$ of a hypergraph $H$, and if $H' := H(V, \calE\setminus F)$, then we can simulate $\CUT_{H'}(\cdot)$ using $\CUT_H(\cdot)$
	and the knowledge of $F$.
\end{observation}

\subsubsection{Proof of Lemma~\ref{lem:mobA}}\label{app:gen-mt}
We prove a generalization of~\Cref{lem:mobA}. To that end, consider the following generalization of~\Cref{def:ghA}.

\begin{definition}\label{def:ghA-em}
	Let $\calX := \{X_1, \ldots, X_t\}$ be a sub-partition of $V$ and let $Y := V \setminus \left(X_1\cup \cdots \cup X_t\right)$. 
	Define
	\[
	g(\calX) := \abs{\{e\in \calE: e\cap X_i \neq \emptyset \forall i,~~e\cap Y\neq \emptyset\}}~~\text{and}~~h(\calX) := \abs{\{e\in \calE: e\cap X_i \neq \emptyset \forall i,~~ e\cap Y =  \emptyset\}}
	\]	
\end{definition}
\noindent
The above definition generalizes $g(X;A)$  and $h(X;A)$. Consider the partition $\calX_{X,A} := \{\{x\}~:~x\in X\} \cup \{A\}$.
Then, first observe that 
\[
g(X;A) = g(\calX_{X,A})~~~\text{and}~~~h(X;A) = h(\calX_{X,A})
\]
since $e\cap \{x\} \neq \emptyset$ is the same as $x\in e$.
Secondly, observe that
the RHS in~\eqref{eq:qXA-em} in the lemma below for $\calX = \calX_{X,A}$ corresponds to the RHS in~\eqref{eq:qXA}.
This can be seen by pairing the subfamily $\calS$ consisting only of singletons with $\calS \cup \{A\}$, which changes the parity 
leading to a term such as $\CUT(S) - \CUT(S\cup A)$. 

\begin{lemma}[Generalized M\"obius Transform]\label{lem:mobA-g}
	For any subpartition $\calX = (X_1, \ldots, X_t)$, define
	\begin{equation}\label{eq:qXA-em}	
		Q(\calX) := \sum_{\emptyset \neq \cS\subseteq \cX} (-1)^{|\cS|+1}~\cdot \underbrace{\CUT\left(\bigcup_{S\in \calS} S\right)}_{\text{call this}~\CUT(\cS)~\text{for brevity}}
	\end{equation}
	Then, we have
	\[
	Q(\calX) = \begin{cases}
		g(\calX) & \textrm{if $|\cX|$ odd} \\
		g(\calX) + 2h(\calX) & \textrm{if $|\cX|$ even}
	\end{cases}
	\]
	Note that $Q(\calX)$ can be computed in $2^{|\calX|+1}$ many $\CUT$-queries.
\end{lemma}
\begin{proof}
	Given a subpartition $\calP$ of $V$, we say $e \cap \calP = \emptyset$ if $e\cap P=\emptyset$ for all $P\in \calP$.
	
	\noindent	
	Fix an edge $e\in \cE$. If $e \cap \cX = \emptyset$ it doesn't participate in any of the terms in the RHS of the definition of $Q(\cX)$. 
	If $e\cap \cX \neq \emptyset$ but there is an $X \in \cX$ with $e\cap X = \emptyset$, then for any $\emptyset \neq \cS\subseteq \cX$ such that $e$ contributes to $\CUT(\cS)$, $e$ also contributes with a different parity to $\CUT(\cS\cup \{X\})$, and thus cancels out.
	So the only $e$'s that survive are ones which intersect {\em all} $X\in \cX$; these are the $e$'s either counted in $g(\cX)$ or $h(\cX)$.
	
	Now, let $Y:= V\setminus \left(\bigcup_{X\in \cX} X\right)$. If $e\cap Y \neq \emptyset$, then it is counted in the RHS exactly $\sum_{\emptyset \neq \cS \subseteq \cX} \left(-1\right)^{|\cS|+1} = 1$ times; this 
	is because of the cancellations of the binomial coefficients except the empty set. Thus, all edges contributing to $g(\cX)$ are counted exactly once. 
	If $e\cap Y = \emptyset$, then the 
	contribution of $e$ to the RHS is precisely $\sum_{\emptyset \neq \cS \subseteq \cX; \cS\neq \cX} \left(-1\right)^{|\cS|+1}$ and this is $0$ if $|\cX|$ is odd and $2$ if $|\cX|$ is even. This implies a contribution of $2h(\cX)$ when $|\cX|$ is even, completing the proof of the lemma.
\end{proof}

\subsection{Even Hypergraphs}
Armed with~\Cref{lem:binsearch-edge-growing}, we can prove the following lemma which, using~\Cref{obs:triv}, immediately proves~\Cref{thm:bnd-even-learn}.
\begin{lemma}
	If $H=(V,\cE)$ is a hypergraph with $|e|$ even for all hyperedges, then given any $x\in V$ and $A\subseteq V\setminus x$, 
	there is an algorithm which either decides there is no edge containing $x$ which intersects $A$, or returns an edge $e$
	with $x\in e$ and $e\cap A \neq \emptyset$. The number of $\CUT$ queries made is $O(2^{|e|}\log n)$.
\end{lemma}

\begin{algorithm}[ht!]
	\caption{Find Edge incident to a vertex in even hypergraph}\label{alg:find-edge-even}
	\begin{varwidth}{\dimexpr\linewidth-2\fboxsep-2\fboxrule\relax}
		\begin{algorithmic}[1]
			\Procedure{FindEdgeEven}{$x,A$}:
			\LineComment{Input: $x\in V$, $A \subseteq V\setminus \{x\}$}
			\LineComment{Assumption: $H$ has {\bf even} hyperedges.}
			\LineComment{Output: $\bot$ or $e$ with $x\in e$, $e\cap A \neq \emptyset$}			
			
			\State $a \eq $ {\sc BinSearch}($\{x\};A$) \Comment{If this returns $\bot$, then there is no desired edge}
			\If{$a = \bot$}:
				\State \Return $\bot$
			\Else:
				\State $X\eq \{x, a\}$ \Comment{Invariant: there is an edge $e$ with $X\subseteq e$}
				\While{True}:
					\State $a\eq ${\sc BinSearch}($X,V\setminus X$)
					\If{$a = \bot$}: \Comment{assert: we have found edge} 
						\State \Return $X$ \label{line:return}
					\Else:
						\State $X\eq X + a$
					\EndIf
				\EndWhile
			\EndIf

			\EndProcedure
		\end{algorithmic}
	\end{varwidth}%
\end{algorithm}
\begin{proof}
	The algorithm is described in~\Cref{alg:find-edge-even}. We begin with {\sc BinSearch}($\{x\};A$); since $|\{x\}|$ is odd, 
	~\Cref{lem:binsearch-edge-growing} implies that if this is $\bot$, there is no edge containing $x$ and intersecting $A$, in which case we return $\bot$.
	Otherwise, {\sc BinSearch} returns an $a\in A$ with the guarantee there is an edge $e \supseteq \{x,a\}$. We now need to find the remaining vertices of one such edge.
	We let $X = \{x,a\}$.
	Since we already found one $a\in A$, we run our binary search {\sc BinSearch}($X, V\setminus X)$. Since $e$ has even sized, we know that either $e = X$ or
	$|e\cap (V\setminus X)| \geq 2$. This will in general be true; if $|X|$ is even and there is an edge $e\supseteq X$, then since $|e|$ is even, either $e = X$ or $|e\cap (V\setminus X)| \geq 2$.
	This proves the correctness of Line~\ref{line:return}. When $|X|$ is odd, then {\sc BinSearch}$(X, V\setminus X)$ will not return $\bot$.
	The total number of $\CUT$ queries made is $O\left(\sum_{i=1}^{|e|}2^i \cdot \log n\right)$ which evaluates to $O(2^{|e|}\log n)$.
\end{proof}
If we also had the ability to make $\EDGE$ queries, then before running Line~\ref{line:return}, if $|X|$ is even, we would actually check
$\EDGE(X\cup y)$ for all $y\in V\setminus X$. If any of them return TRUE, then we return that edge. Otherwise, we return $X$. 
This is correct because {\sc BinSearch}($X,V\setminus X) = \bot$ implies there is no edge $e$ with $X\subseteq e$ and $|e\cap (V\setminus X)| \geq 2$. However, 
the invariant is that there is an edge with $X\subseteq e$. So either $e = X$ or $e = X + y$ for some $y\in V\setminus X$.
So the above algorithm works in any hypergraph but with $n$ additional $\EDGE$ queries. This immediately proves~\Cref{thm:xtra-cut+edge-learn}.

Armed with the above algorithm, we can easily find a spanning subhypergraph using a Prim-style approach described in~\Cref{alg:hss}.

\begin{algorithm}[h!]
	\caption{Finding Spanning Subhypergraph}\label{alg:hss}
	\begin{varwidth}{\dimexpr\linewidth-2\fboxsep-2\fboxrule\relax}
	\begin{algorithmic}[1]
		\Procedure{FindSpanSubHG}{$V,\CUT, x\in V$}:
			\LineComment{Input: $n$ vertices $V$ and $\CUT$ query access to hidden {\bf even} hypergraph $H$}
			\LineComment{Output: spanning subhypergraph $H'$ of component containing $x$}			
			
			\State $C\eq \{x\}$; $U \eq V\setminus C$; $\calE'\eq \emptyset$ 
			\State $D\eq \emptyset$ \Comment{$D$ is set of ``dead'' vertices; they have no edges containing them and touching $U$}
			\While{$D\neq C$}:
				\State Pick $x\in C\setminus D$ and run $e \eq ${\sc FindEdgeEven}($x, U$)
				\If{$e = \bot$}:
					\State $D \eq D + x$
					\State Continue to next while loop
				\Else: \Comment{$e$ is an edge containing $x$ and some vertex in $U$}
					\State $\cE' \eq \cE' \cup \{e\}$; $C\eq C \cup e$; $U\eq V\setminus C$
				\EndIf
			\EndWhile
			\State \Return $H' \eq (C, \cE')$.
		\EndProcedure
	\end{algorithmic}
\end{varwidth}%
\end{algorithm}
\noindent
The time taken is $O(\sum_{e\in \calE'} 2^{|e|}\log n)$; if $|e|\leq r$, then since $|\calE'| \leq n-1$, we obtain the spanning sub-hypergraph in $O(2^r n\log n)$ many $\CUT$ queries.
 Along with~\Cref{thm:ni}, this establishes~\Cref{thm:bnd-even-sparse}.

\subsection{Linear Hypergraphs}

Recall, a hypergraph $H = (V, \cE)$ is {\em linear} if for any two different $e,f \in \cE$, we have $|e\cap f|\leq 1$. Furthermore, let us call a hypergraph large if $|e| \geq 4$ for every $e\in \calE$.
We first show that large, linear hypergraphs can be reconstructed easily. The main place where linearity is used is the following: if we know there exists an edge $e$ such that $\{x,y\} \subseteq e$, 
then there is {\em exactly} one such edge $e^*$. 
So, if we find that there is an edge $e$ containing $\{x,y,z\}$, it must be the same edge $e^*$; if we find there is an edge $e$ containing $\{x,y,z'\}$, it must be the same edge $e^*$. The algorithm is now clear.
Since $|e^*|\geq 4$, we first use \Cref{lem:binsearch-edge-growing}
with $X=\{x,y\}$ and $A=V\setminus X$ to find a vertex $z$ with
$\{x,y,z\}\subseteq e^*$. We then fix the odd anchor
$X:=\{x,y,z\}$. Repeatedly applying \Cref{lem:binsearch-edge-growing}
to $(X,A)$ finds all remaining vertices of the unique edge containing
$X$; uniqueness follows from linearity. When the search returns $\bot$,
there is no further vertex of that edge.
The full description of the algorithm is in~\Cref{alg:e-llh}.
\begin{algorithm}[!h]
	\caption{Finding one edge in large linear hypergraph}\label{alg:e-llh}
	\begin{varwidth}{\dimexpr\linewidth-2\fboxsep-2\fboxrule\relax}
		\begin{algorithmic}[1]
			\Procedure{FindHypEdgeLL}{$V,\CUT, x\in V$}:
			\LineComment{Input: $n$ vertices $V$ and $\CUT$ query access to hidden hypergraph $H$}
			\LineComment{Assumption: $|e\cap f|\leq 1$ for all distinct $e,f\in \cE$ and $|e| \geq 4$ for $e\in \cE$}
			\LineComment{Assumption: $x$ has at least an edge incident on it.}
			\LineComment{Output: some edge of $H$ incident on $x$}
			\State Use {\sc BinSearch}($\{x\}, V\setminus\{x\})$ from~\Cref{lem:binsearch-edge-growing} to find $y$.
			\LineComment{Linearity asserts there exists {\bf unique} edge $e^*$ with $\{x,y\}\subseteq e^*$.}
			\State Use {\sc BinSearch}($\{x,y\},V\setminus \{x,y\}$) from~\Cref{lem:binsearch-edge-growing}  to obtain $z$. \label{alg:ll:l8}
			\LineComment{Since $|e^*|\geq 4$ we know $|e^* \cap (V\setminus \{x,y\})|\geq 2$ implying $\{x,y,z\} \subseteq e^*$.}
			\State $e \eq \{x,y,z\}$; $X\eq \{x,y,z\}$; $A\eq V\setminus X$
			\While{True}:
			\State $z \eq$ {\sc BinSearch}($X,A$)
			\LineComment{Since $|X|$ is odd, this is guaranteed to give other elements of $e^*$ until $\bot$}
			\If{$z=\bot$}
			\State abort while loop
			\Else
			\State $e\eq e\cup \{z\}$; $A\eq A\setminus\{z\}$
			\EndIf
			\EndWhile
			\State \Return $e$
			\EndProcedure
		\end{algorithmic}
	\end{varwidth}%
\end{algorithm}
\begin{remark}
	The above algorithm also works if $|e|=2$ or $|e| \geq 4$ for every $e\in \calE$. 
	The only difference is that Line 8 of~\Cref{alg:e-llh}
	may return $\bot$ in which case $\{x,y\}$ is the unique $2$-edge $e^*$ containing both $x$ and $y$.
\end{remark}

Since the size of $|X|$ for which $Q(X;A)$ is called is always $\leq 3$, the above algorithm makes $O(|e^*|\log n)$-many $\CUT$ queries.
Armed with~\Cref{obs:triv} (and noting that deleting an edge from a large, linear hypergraph maintains those properties), we obtain 
an algorithm which reconstructs all the edges in $O(\sum_{e\in \cE} |e| \cdot \log n)$ many $\CUT$ queries, proving~\Cref{thm:linear-learn}.  To obtain a spanning subhypergraph, we use the same ``Prim-style'' ideas as employed in~\Cref{alg:hss}.  The following claim relates the ``size'' $p_H$ of linear hypergraphs.

\begin{claim}[Size of Linear Hypergraphs]\label{clm:size-linear}
	If $H = (V, \cE)$ is a linear hypergraph with $n = |V|$, $m = |\cE|$ and $p = \sum_{e\in \cE} |e|$, then 
	(i) $m \leq \binom{n}{2}$, and (ii) $p = O(n\sqrt{m}) = O(n^2)$.
\end{claim}
\begin{proof}
	If we count the set $\{(x,y,e): x\in V, y\in V, e\in \cE, \{x,y\}\subseteq e\}$ and use linearity, we get that 
$\sum_{e\in \cE} \binom{|e|}{2} \leq \binom{n}{2}$. The LHS is at least $|\cE| = m$ proving (i). Cauchy-Schwarz proves (ii)
since the aforementioned inequality implies $\sum_{e\in \cE} |e|^2 = O(n^2)$, and is at least $p^2/m$.
\end{proof}
\noindent
A spanning subhypergraph $H'$ has at most $n$ edges, and thus $p_{H'} = O(n^{1.5})$ showing that we can find a spanning subhypergraph in 
a large, linear hypergraph in $O(n^{1.5}\log n)$-many $\CUT$ queries. Next, we discuss how to remove the largeness assumption. 
Before moving on to do so, we state the obvious generalization which may be of interest: it weakens the linearity to $t$-linearity, but then requires
all edges to be much bigger; one uses H\"older instead of Cauchy-Schwarz.

\begin{theorem}\label{thm:xtra-t-linear}
For any $t \in \NN$, if a hypergraph $H = (V,\cE)$ satisfies (i) $|e|\geq t+3$ for all $e\in \cE$, 
and (ii) $|e\cap f| \leq t$ for all distinct $e,f\in E$, then a spanning subhypergraph can be learned
in $O(n^{2 - \frac{1}{(t+1)}}\log n)$-many $\CUT$-queries.
%
\end{theorem}

\def\tH{\widetilde{H}}
\paragraph{Removing the largeness assumption.} From the above discussion, we can see that one could learn all large hyperedges, but one is stuck at 
edges $e$ with $|e| \leq 3$. And indeed, as noted in \Cref{sec:intro}, one cannot exploit linearity to learn the hypergraph exactly. 
However, one other simple observation still allows us to construct ``sparse certificates'' --- we put quotes because we won't be able to find 
$3$-edges of $H$ -- but we don't need to; a triangle is the same as three normal edges with ``half'' the weight.
\begin{observation}
	Let $H = (V,\cE)$ be any hypergraph. Construct $\tH = (V,\widetilde{\cE})$ which is a hypergraph
	formed as follows: for each $e\in \cE$ with $|e| = 2$ or $|e|\geq 4$, we add {\em two} copies of $e$ to $\widetilde{\cE}$, 
	and for $e = \{a,b,c\}$ in $\cE$, we add $(a,b), (b,c)$ and $(c,a)$ to $\widetilde{\cE}$.	Then, for any $S\subseteq V$, 
	$\CUT_{\tH}(S) = 2\CUT_H(S)$. That is, $\tH$ and $2H$ are cut-equivalent. Furthermore, given any $\widetilde{e} \in \widetilde{\cE}$, 
	one can simulate $\CUT_{\tH - \widetilde{e}}(S)$ as $2\CUT_H(S) - 1$ if $S\cap e \neq \emptyset$ and $(V\setminus S) \cap e \neq \emptyset$, 
	and $2\CUT_H(S)$ otherwise.
\end{observation}
\noindent
Since $\CUT()$ is additive over the hyperedges, it suffices to argue for a single edge $e \in H$ with $|e| = 3$. Now one can check this brute-force; the main reason is that
any relevant cut looks the same: one vertex of $e$ on one side and two on the other. This is, unfortunately, particular to $|e| = 3$, and we are not able to generalize to $|e| = 5$ or other odd hyperedges. More precisely, if we could replace 
all odd-sized hyperedges with even-sized hyperedges and be cut-equivalent, then we would be able to generalize~\Cref{rslt:bnd-even} to all bounded hypergraphs. 

Armed with the above observation, we can obtain~\Cref{rslt:lin} as follows. We assume we are working with $\tH$ and since there are no $|e| = 3$, 
\Cref{alg:e-llh} (or rather its spanning subhypergraph version) will find a spanning subhypergraph $H_1$ of $H$ in $O(n^{1.5}\log n)$ many $\CUT$ queries.
Therefore, we can apply~\Cref{thm:ni} to find $H_1, \ldots, H_{2k}$ such that their union is a $2k$-sparse certificate of $\tH$. 
Since $\tH$ and $2H$ are cut-equivalent, this union is a $k$-sparse certificate (but not a subhypergraph) of $H$. 
\Cref{rslt:lin} is established fully by noting that in $O(n^2\log n)$ many $\CUT$-queries, we can learn all of $\tH$.

\appendix
\section{Odd Hypergraphs with Same Cut Profile}
\label{sec:app?}

In this section we show that odd parity of cardinality of edges is a key hindrance to learning hypergraphs. In particular, we prove the following theorem.
	\begin{theorem}
		Let $k\geq 3$ be an odd number. There are two $k$-uniform hypergraphs $H_0$ and $H_1$ on the same labeled vertex set $V$ whose
		cut-profiles are the same.
	\end{theorem}
	\begin{proof}
	The vertex set $V$ has $2k$ vertices named $v_1, \ldots, v_k$ and $\bar{v}_1, \ldots, \bar{v}_k$.
	The set of edges $e\in H_0$ (and in $\in H_1$, respectively) are all subsets of $V$ such that $e$ contains {\em exactly} one
	vertex from $\{v_i, \bar{v}_i\}$ for $1\leq i\leq k$ {\em and} the number of $\bar{v}_i$'s is {\em even} (respectively, {\em odd}).
	Thus, every edge has exactly $k$ vertices. 
	The number of edges in $H_0$ is $\sum_{0\leq i\leq k: i~\text{even}} \binom{k}{i}$ and the number of edges in $H_1$ is
	$\sum_{0\leq i\leq k: i~\text{odd}} \binom{k}{i}$. Since $k$ is odd, both expressions are the same and equal to $2^{k-1}$.

	Let $\CUT_0$ and $\CUT_1$ be cut-functions in $H_0$ and $H_1$, respectively. We want to argue that $\CUT_0(S) = \CUT_1(S)$ for any $S\subseteq V$. 
	It will be convenient to count the number of edges {\em not counted} in $\CUT_i(S)$ for $i\in \{0,1\}$; in particular, define $g_i(S)$ to be the number 
	of edges $e\in E_i$ such that $e\subseteq S$ or $e\subseteq V\setminus S$. So, $\CUT_i(S) = 2^{k-1} - g_i(S)$, and so it suffices to show $g_0(S) = g_1(S)$.
	
	Call $S$ {\em consistent} if $|S\cap \{v_i, \bar{v}_i\}| = 1$ for all $1\leq i\leq k$.
	Define a map  $\phi_S : E(H_0) \to E(H_1)$ by {\em flipping the bits} as follows: for $e\in E(H_0)$, 
	we define $\phi(e)$ to contain $v_i$ if $\bar{v}_i \in e$ or vice-versa. 
	Note that if the number of ``barred'' vertices in $e$ is $j$, then the number of ``barred vertices'' in $\phi_S(e)$ is $k-j$. Since $k$ is odd, we get that $\phi_S(e) \in E(H_1)$.
	It is easy to see this is a bijection. Furthermore, note that
	$e \subseteq S \Leftrightarrow \phi_S(e) \subseteq V\setminus S$ and $e\subseteq V\setminus S \Leftrightarrow \phi_S(e) \subseteq S$.
	And so there is a bijection between the set of uncut edges proving $g_0(S) = g_1(S)$ for consistent subsets.
	
	Let $S$ be an {\em inconsistent} subset, and let $i$ be the smallest index such that $\{v_i, \bar{v}_i\} \subseteq S$ (the case of $\subseteq V\setminus S$ would be analogous.).
	Now note that no edge $e$, either in $H_0$ or $H_1$, can be a subset of $V\setminus S$ (in the analogous case, subset of $S$). Now, for $e\in E(H_0)$, define
	$\phi_S(e)$ to be the edge where the parity of $v_i$ is flipped; that is, if $v_i\in e$ then $\bar{v}_i \in e$ and everything else remains the same.
	It is easy to see this is a bijection. Furthermore, if $e\subseteq S$ then $\phi_S(e) \subseteq S$. And this shows that $g_0(S) = g_1(S)$ for such 
	an inconsistent subset as well. This completes the proof.
	\end{proof}
\noindent
When $k=3$, the above two hypergraphs are indeed {\em linear}; in fact, every pair of hyperedges intersect in exactly one vertex.

\bibliographystyle{plain}
\bibliography{main}

@String { ACM             = {The OX Association for Computing Machinery} }

@String { alt             = {Proc., International Conference on Algorithmic Learning Theory (ALT)} }

@String { colt            = {Proc., Conf. on Learning Theory (COLT)} }

@String { esa             = {Proc., European Symposium on Algorithms (ESA)} }

@String { focs            = {Proc., IEEE Conference on the Foundations of Computer Science (FOCS)} }

@String { fsttcs          = {Proc., FSTTCS} }

@String { icalp           = {Proc., International Conference on Algorithms, Logic, and Programming (ICALP)} }

@String { itcs           = {Proc., Innovations in Theoretical Computer Science (ITCS)} }

@String { pods            = {Proc., ACM Symposium on Principles of Database Systems (PODS)} }

@String { sidma           = {SIAM Journal on Discrete Mathematics (SIDMA)} }

@String { soda            = {Proc., ACM-SIAM Symposium on Discrete Algorithms (SODA)} }

@String { stoc            = {Proc., ACM Symposium on the Theory of Computing (STOC)} }

@article{GaoC22,
	title={A Linear Hypergraph Extension of Tur{\'a}n's Theorem},
	author={Gao, Guorong and Chang, An},
	journal={The Electronic Journal of Combinatorics},
	pages={P4--41},
	year={2022}
}

@article{ImL25,
	title={Dirac’s theorem for linear hypergraphs},
	author={Im, Seonghyuk and Lee, Hyunwoo},
	journal=sidma,
	volume={39},
	number={2},
	pages={834--847},
	year={2025}
}

@article{JanzerST24,
	title={Regular subgraphs of linear hypergraphs},
	author={Janzer, Oliver and Sudakov, Benny and Tomon, Istv{\'a}n},
	journal={International Mathematics Research Notices},
	volume={2024},
	number={17},
	pages={12366--12381},
	year={2024}
}

@inproceedings{JiangNS26,
	title={Minimum Cuts with Fewer Cut Queries},
	author={Jiang, Yonggang and Nanongkai, Danupon and Sawettamalya, Pachara},
	booktitle=soda,
	pages={258--296},
	year={2026}
}

@inproceedings{KennethK26a,
	title={All-Pairs Minimum Cut using Cut Queries},
	author={Kenneth-Mordoch, Yotam and Krauthgamer, Robert},
	booktitle=soda,
	pages={4077--4095},
	year={2026}
}

@inproceedings{KennethK26b,
	title={Faster All-Pairs Minimum Cut: Bypassing Exact Max-Flow},
	author={Kenneth-Mordoch, Yotam and Krauthgamer, Robert},
	booktitle=stoc,
	pages={to appear},
	year={2026}
}

@inproceedings{RubinsteinSW18,
  author={Aviad Rubinstein and Tselil Schramm and S. Matthew Weinberg},
  title={Computing Exact Minimum Cuts Without Knowing the Graph},
  booktitle=itcs,
  pages={39:1--39:16},
  year={2018}
}

@inproceedings{MukhopadhyayN20,
  title={Weighted min-cut: sequential, cut-query, and streaming algorithms},
  author={Mukhopadhyay, Sagnik and Nanongkai, Danupon},
  booktitle=stoc,
  pages={496--509},
  year={2020}
}

@article{ApersEGLMN22,
  author={Simon Apers and Yuval Efron and Pawel Gawrychowski and Troy Lee and Sagnik Mukhopadhyay and Danupon Nanongkai},
  title={Cut query algorithms with star contraction},
  journal=focs,
  pages = {507--518},
  year={2022}
}

@inproceedings{Choi13,
  author={Sung{-}Soon Choi},
  title={Polynomial Time Optimal Query Algorithms for Finding Graphs with Arbitrary Real Weights},
  booktitle = colt,
  volume={30},
  pages={797--818},
  year={2013}
}

@inproceedings{AssadiCK21,
  title={Graph Connectivity and Single Element Recovery via Linear and OR Queries},
  author={Assadi, Sepehr and Chakrabarty, Deeparnab and Khanna, Sanjeev},
  booktitle=esa,
  year={2021}
}

@inproceedings{LeeMS21,
  title={Quantum algorithms for graph problems with cut queries},
  author={Lee, Troy and Santha, Miklos and Zhang, Shengyu},
  booktitle=soda,
  pages={939--958},
  year={2021}
}

@article{AuzaL21,
  title={On the query complexity of connectivity with global queries},
  author={Auza, Arinta and Lee, Troy},
  journal={arXiv preprint arXiv:2109.02115},
  year={2021}
}

@inproceedings{ChekuriQ21,
  title={Isolating cuts, (bi-)submodularity, and faster algorithms for connectivity},
  author={Chekuri, Chandra and Quanrud, Kent},
  booktitle=icalp,
  pages = {50:1--50:20},
  year={2021}
}

@article{Queyranne98,
  title={Minimizing symmetric submodular functions},
  author={Queyranne, Maurice},
  journal={Mathematical Programming},
  volume={82},
  pages={3--12},
  year={1998}
}

@phdthesis{Harvey08,
  title={Matchings, matroids and submodular functions},
  author={Harvey, Nicholas James Alexander},
  year={2008},
  school={Massachusetts Institute of Technology}
}

@InProceedings{ChakrabartyL23,
  title={A Query Algorithm for Learning a Spanning Forest in Weighted Undirected Graphs},
  author={Chakrabarty, Deeparnab and Liao, Hang},
  booktitle=alt,
  pages={259--274},
  year={2023}
}

@inproceedings{LiaoC24,
  author={Liao, Hang and Chakrabarty, Deeparnab},
  title={Learning Spanning Forests Optimally in Weighted Undirected Graphs with CUT queries},
  pages={785--807},
  booktitle=alt,
  year={2024}
}

@inproceedings{ChakrabartyL24,
  author={Chakrabarty, Deeparnab and Liao, Hang},
  title={Learning Partitions Using Rank Queries},
  pages={16:1--16:14},
  booktitle=fsttcs,
  year={2024}
}

@inproceedings{AngluinC05,
  title={Learning a hidden hypergraph},
  author={Angluin, Dana and Chen, Jiang},
  booktitle=colt,
  pages={561--575},
  year={2005},
  organization={Springer}
}

@inproceedings{ChenKN20,
  title={Near-linear size hypergraph cut sparsifiers},
  author={Chen, Yu and Khanna, Sanjeev and Nagda, Ansh},
  booktitle=focs,
  pages={61--72},
  year={2020}
}

@inproceedings{ChenKN21,
  title={Sublinear Time Hypergraph Sparsification via Cut and Edge Sampling Queries},
  author={Chen, Yu and Khanna, Sanjeev and Nagda, Ansh},
  booktitle=icalp,
  pages = {53:1--53:21},
  year={2021}
}

@article{ChekuriX17,
	title={Minimum cuts and sparsification in hypergraphs},
	author={Chekuri, Chandra and Xu, Chao},
	journal={SIAM Journal on Computing},
	volume={47},
	number={6},
	pages={2118--2156},
	year={2018}
}

@InProceedings{KennethK23,
	author =	{Kenneth, Yotam and Krauthgamer, Robert},
	title =	{{Cut Sparsification and Succinct Representation of Submodular Hypergraphs}},
	booktitle =	icalp,
	pages =	{97:1--97:17},
	year =	{2024}
}

@inproceedings{Quanrud24,
  title={Quotient sparsification for submodular functions},
  author={Quanrud, Kent},
  booktitle=soda,
  pages = {5209--5248},
  year={2024}
}

@inproceedings{LiP20,
  title={Deterministic min-cut in poly-logarithmic max-flows},
  author={Li, Jason and Panigrahi, Debmalya},
  booktitle=focs,
  pages={85--92},
  year={2020}
}

@article{NagamochiI92,
  title={A linear-time algorithm for finding a sparse k-connected spanning subgraph of a $k$-connected graph},
  author={Nagamochi, Hiroshi and Ibaraki, Toshihide},
  journal={Algorithmica},
  volume={7},
  number={1},
  pages={583--596},
  year={1992},
  publisher={Springer}
}

@inproceedings{GuhaMT15,
  title={Vertex and hyperedge connectivity in dynamic graph streams},
  author={Guha, Sudipto and McGregor, Andrew and Tench, David},
  booktitle=pods,
  pages={241--247},
  year={2015}
}

@inproceedings{AnandSW25,
  title={Deterministic Edge Connectivity and Max Flow using Subquadratic Cut Queries},
  author={Anand, Aditya and Saranurak, Thatchaphol and Wang, Yunfan},
  booktitle={Proceedings of the 2025 Annual ACM-SIAM Symposium on Discrete Algorithms (SODA)},
  pages={124--142},
  year={2025}
}

@InProceedings{KennethK25,
	author =	{Kenneth-Mordoch, Yotam and Krauthgamer, Robert},
	title =	{{Cut-Query Algorithms with Few Rounds}},
	booktitle =	esa,
	pages =	{100:1--100:14},
	year =	{2025}
}

@article{Jiang22,
  title={Minimizing convex functions with rational minimizers},
  author={Jiang, Haotian},
  journal={Journal of the ACM},
  volume={70},
  number={1},
  pages={1--27},
  year={2022},
  publisher={ACM New York, NY}
}

@inproceedings{ChakrabartyGJS22,
  title={Improved lower bounds for submodular function minimization},
  author={Chakrabarty, Deeparnab and Graur, Andrei and Jiang, Haotian and Sidford, Aaron},
  booktitle=focs,
  pages={245--254},
  year={2022}
}

@inproceedings{KoganK15,
  title={Sketching cuts in graphs and hypergraphs},
  author={Kogan, Dmitry and Krauthgamer, Robert},
  booktitle={Proceedings of the 2015 Conference on Innovations in Theoretical Computer Science},
  pages={367--376},
  year={2015}
}

@inproceedings{KhannaPS24,
  title={Near-optimal size linear sketches for hypergraph cut sparsifiers},
  author={Khanna, Sanjeev and Putterman, Aaron and Sudan, Madhu},
  booktitle=focs,
  pages={1669--1706},
  year={2024}
}
\end{document}